\documentclass[a4paper]{easychair}
\usepackage{times}
\usepackage{amsmath,amsfonts,amssymb,amsthm}
\usepackage[english]{babel}
\usepackage{listings,url}
\usepackage{verbatim}
\usepackage{mathpartir}
\usepackage{epsfig}
\usepackage{proof}
\usepackage{tikz}
\usepackage{pgfplots}
\usepackage{paralist}
\usepackage{ifdraft}
\usepackage{usefuldefs}
\usepackage[utf8]{inputenc}
\usetikzlibrary{arrows,automata,patterns,calc}
\urldef{\mailsa}\path|{schwmart,seidl}@in.tum.de|
\urldef{\mailsb}\path|thomas.martin.gawlitza@uni-oldenburg.de|
\renewcommand{\qed}{\hfill\mbox{\rule[0pt]{1.3ex}{1.3ex}}}

\newcommand{\NN}{\mathbb{N}}

\newcommand{\ZZ}{\mathbb{Z}}
\newcommand{\ZzZ}{\overline{\mathbb Z}}

\newenvironment{mtx*}{\begin{smallmatrix}}{\end{smallmatrix}}

\newcommand{\D}{{\mathbb D}}

\newcommand\justification[1]{\quad[\mbox{#1}]}

\catcode`\@=11        
\def\ignoretrue{\global\@ignoretrue}
\catcode`\@=12        

\def\defrelations{
   \def\better{\etc{\geq}}         
   \def\worse {\etc{\leq}}         %
   \def\equals{\etc{=}}            %
   \def\equivals{\etc{\equiv}}     %
   \def\lesseq{\etc{\leq}}         %
   \def\greateq{\etc{\geq}}        %
   \def\subseteqwithoutwhy{\etc{\subseteq}}
   \def\superseteq{\etc{\supseteq}}    %
   \def\inwithoutwhy{\etc{\in}}    %
   \def\iff{\etc{\mbox{iff}}}      
   \def\betterwhy##1{\etcwhy{\geq}{##1}}
   \def\worsewhy##1{\etcwhy{\leq}{##1}}
   \def\equalswhy##1{\etcwhy{=}{##1}}
   \def\equivalswhy##1{\etcwhy{\equiv}{##1}}
   \def\lesseqwhy##1{\etcwhy{\leq}{##1}}
   \def\greateqwhy##1{\etcwhy{\geq}{##1}}
   \def\subseteqwhy##1{\etcwhy{\subseteq}{##1}}
   \def\supseteqwhy##1{\etcwhy{\supseteq}{##1}}
   \def\sqsubseteqwhy##1{\etcwhy{\sqsubseteq}{##1}}
   \def\sqsupseteqwhy##1{\etcwhy{\sqsupseteq}{##1}}
   \def\iffwhy##1{\etcwhy{\mbox{iff}}{##1}}
   \def\implieswhy##1{\etcwhy{\Rightarrow}{##1}}
   \def\inwhy##1{\etcwhy{\in}{##1}}
   \def\LHS{{\rm LHS}}              
   \def\RHS{{\rm RHS}}              
}

{  \def\termskip{1ex}             
   \def\commentskip{0.5ex}        
   \def\etc##1{\cr\noalign{\vspace{\termskip}}##1&}
   \def\etcwhy##1##2{\cr\noalign{\vspace{\termskip}}##1&\justification{##2}\cr\noalign{\vspace*{\commentskip}}&}
    \bgroup
    \defrelations
    \tabskip=1em
    $$\halign to \displaywidth \bgroup\strut
        $\displaystyle\tabskip0pt{##}$\tabskip1ex&$\displaystyle\tabskip0pt{##}$\tabskip=0pt plus \textwidth&\tabskip=0pt plus 0pt\llap{##}\cr&}%
{\cr\egroup$$\egroup\ignoretrue}

{  \def\tabsign{&}
   \def\signspace{\ \ }  
   \def\etc##1{\tabsign\signspace ##1\signspace\tabsign\def\tabsign{}}
   \def\etcwhy##1##2
       { \tabsign\signspace
         \begin{array}[t]{@{}c@{}}
           ##1 \\[-0.4ex] \uparrow \\ \makebox[0mm]{[##2]}
         \end{array}
         \signspace
         \tabsign
         \def\tabsign{}
       }
   \def\linecalcendsequence{}
   \def\qed{\gdef\linecalcendsequence{\hfill\proofendsign}}

   \defrelations
   \[ \begin{array}[t]{@{}l@{}c@{}l@{}}}%
{     \end{array} \]\linecalcendsequence\ignoretrue}%

{  \[
   \begin{array}{l@{~~~}l@{~}c@{~}l@{~~}l}}%
{  \end{array}\]\ignoretrue}%

\usepackage{xspace}

\allowdisplaybreaks

\newcommand{\us}{\cite{Gawlitza07PreciseFix,DBLP:conf/fm/GawlitzaS08,EasyChair:235}\xspace}

\begin{document}
\newtheorem{example}{Example}
\newtheorem{theorem}{Theorem}
\newtheorem{lemma}{Lemma}

\title{Parametric Strategy Iteration}
\titlerunning{Parametric Strategy Iteration}
\authorrunning{T.M. Gawlitza, M.D. Schwarz, H. Seidl}

\author{Thomas M. Gawlitza\inst{1} \and Martin D. Schwarz\inst{2} \and Helmut Seidl\inst{2}}
\institute{
Carl von Ossietzky Universität Oldenburg, 
Ammerländer Heerstraße 114-118, D-26129 Oldenburg, Germany 
\mailsb
\and
Technische Universit\"at M\"unchen, 
Boltzmannstra\ss e 3, D-85748 Garching, Germany
\mailsa
}


\maketitle


\begin{abstract}
Program behavior may depend on parameters, which are either configured
before compilation time, or provided at runtime, e.g., by sensors or other input devices.
Parametric program analysis explores how different parameter settings may affect the
program behavior.

In order to infer invariants depending on parameters, we introduce parametric strategy iteration.
This algorithm determines the precise least solution of systems of integer equations depending
on surplus parameters. Conceptually, our algorithm performs ordinary strategy iteration 
on the given integer system for all possible parameter settings in parallel.
This is made possible by means of region trees to represent the occurring piecewise affine functions.
We indicate that each required operation on these trees is polynomial-time if only constantly many 
parameters are involved.

Parametric strategy iteration for systems of integer equations
allows to construct parametric integer interval analysis
as well as parametric analysis of differences of integer variables.
It thus provides a general technique to realize precise parametric
program analysis if numerical properties of integer variables are of concern.
\end{abstract}


\section{Introduction} \label{s:intro}

Since the very beginnings of linear programming, 
\emph{parametric versions} of linear programming (LP for short) 
have already been of interest
(see, e.g., \cite{feautrier88,gal97} 
for an overview and \cite{holder11} 
for recent algorithms).
Parametric LP can be applied to answer questions
such as: how much does the result of the analysis (optimal value/solution) depend on specific parameters?
What is the precise dependency between a parameter and the result?
In which regions of parameter values do these dependencies significantly change?
Such types of \emph{sensitivity} and \emph{mode} analyses are important in order to obtain a better
understanding of the problem under consideration and its potential 
analysis.
Sensitivity and mode questions equally apply to
programs whose behavior depends on parameters. 
Such parameters either could be provided at configuration time by engineers,
or at runtime, e.g., through sensor data or other kinds of input. 
The goal then is to determine how the output values produced by the program
may be influenced by the parameters.
The same software may, e.g., control the break of a truck or a car, 
but must behave quite differently in the two use cases.


Here, we consider the static analysis of parameterized systems and
propose methods for inferring, how numerical program invariants
may depend on the parameters of the program.
These questions cannot be answered by linear or integer linear programming related techniques alone,
since the constraints to be solved are necessarily non-convex.
Still, such questions 
can be answered for interval analysis or,
more generally, for template-based numerical analyses such as difference bound matrices
or octagons, since their analysis results can be expressed in \emph{first-order linear real arithmetic}.
This observation has been exploited by Monniaux \cite{monniaux09} 
who applied 
\emph{quantifier elimination algorithms} to obtain parametric analysis results.
The resulting system can provide amazing results. 
However, 
if the programs under consideration have more \emph{complicated control-flow}, i.e., do not 
consist of a single program point only, 
fixpoint computation realized by means of quantifier elimination does no longer seem appropriate. 

In this paper, we are concerned with invariants over integers (opposed to rationals in Monniaux'
system).
%
\begin{example}\label{e:par}\rm
Consider the following parametric program:
	\begin{align*}
	& x = {\bf p}_1;	\;
	{\bf while}\;(x<{\bf p}_2)	\;
	x = x+1;
	\end{align*}

\noindent
where ${\bf p}_1,{\bf p}_2$ are the parameters. 
The parametric invariant for program exit 
states that $x ={\bf p}_1$ holds, if ${\bf p}_2 \leq {\bf p}_1$,
and $x={\bf p}_2$ otherwise.
Thus, the analysis should distinguish two modes where the invariant inferred for 
program exit is significantly different, namely the set of parameter settings 
where ${\bf p}_2 \leq {\bf p}_1$ holds and its complement.
In the first mode, the value of $x$ at program exit is only sensitive to changes of the
parameter ${\bf p}_1$, while otherwise it is sensitive to changes of the
parameter ${\bf p}_2$.
\qed
\end{example}

\noindent
Note that for \emph{integer linear arithmetic}, quantifier elimination
is even more intricate than over the rationals.
%
As shown in \cite{Gawlitza07PreciseFix,DBLP:conf/fm/GawlitzaS08,EasyChair:235}, non-parametric
interval analysis as well as the analysis of differences of integer variables can be compiled into 
suitable integer equations. 
In absence of parameters, integer equation systems where no integer multiplication or division is involved,
can be solved without resorting to heavy machinery such as \emph{integer linear programming}. 
Instead, an iteration over \emph{max}-strategies suffices \us. Here, a \emph{max}-strategy maps each application of a maximum operator to one of its arguments.
Once a choice is made at each occurrence of a maximum operator, a conceptually simpler system is obtained.
For systems without  maximum operators, the \emph{greatest} solution can be determined by means of 
a generalization of the \emph{Bellman-Ford} algorithm. 
%
The greatest solutions of the maximum-free systems
encountered during the iteration on \emph{max}-strategies, provide us with
an increasing sequence of \emph{lower} approximations to the overall least solution of the integer system.
Given such a lower approximation, we can check whether a solution and thus the least solution
has already been reached. Otherwise, the given \emph{max}-strategy is improved, and the iteration proceeds.

In contrast to ordinary program analysis,
\emph{parametric program analysis} infers a distinct program invariant 
for each possible parameter setting.
To solve these parametric systems,
we propose to apply strategy iteration \emph{simultaneously} for all parameter settings 
(Section \ref{s:psi}).
We show that this algorithm terminates---given that we can effectively deal
with the parametric intermediate results considered by the algorithm.
For that, we show that
the intermediate parametric values can be represented 
by \emph{region trees} (Section \ref{s:region}).
Here, a region tree over a finite set $C$ of linear inequalities is a data-structure
for representing finite partitions of the parameter space into non-empty regions of parameter settings which
are indistinguishable by means of the constraints in $C$. 
A value (of a set $V$ of values) is then assigned to each non-empty region of the tree. 
We also indicate that
each basic operation which is required for implementing parametric strategy iteration 
can be realized with region trees in polynomial time --- assuming that the number of parameters is fixed
(Section \ref{s:implementation}).
Finally, we apply parametric strategy iteration for parametric integer equations
to solve parametric interval equations and thus to perform
parametric program analysis
for integer variables of programs, and report preliminary experimental results 
(Section \ref{s:experiments}).

\section{Basic Concepts}\label{s:basics}

In this section we provide basic notions 
and introduce systems of \emph{parametric integer equations}.
By $\ZzZ$ we denote the complete linearly ordered set $\ZZ \cup \{\neginfty,\infty\}$.
%
%
Let ${\bf P}$ and ${\bf X}$ denote finite sets of \emph{parameters} and \emph{variables} 
(\emph{unknowns}), which are disjoint.
A system $\E$ of \emph{parametric} integer equations is given by
	\begin{align*}
	{\bf x} = e_{\bf x}\;,\qquad{\bf x}\in{\bf X}
	\end{align*}
where each right-hand side $e_{\bf x}$ is of the form $e_1\vee\ldots\vee e_r$
for parametric integer expressions $e_1,\ldots,e_r$. 
Here, ``$\vee$'' denotes the maximum operator.
In this paper, an integer expression is built up from constants
in $\ZzZ$, variables, parameters, and negated parameters by means of application of operators.
As operators, we consider ``$\wedge$'' (minimum), ``$+$'' (addition),
``$;$'' (test of non-negativity) and multiplication with non-negative scalars.
Addition as well as scalar multiplication is extended to $\neginfty$ and $\infty$ by:
\begin{align*}
x + \neginfty &= \neginfty + x = \neginfty	&&\forall x\in\ZzZ	\\
x + \infty    &= \infty + x 	= \infty	&&\forall x\in\ZzZ\backslash\{\neginfty\}\\
c\cdot(\neginfty) &= \neginfty		&&\forall c\geq 0	\\
0\cdot\infty	&= 0			&&\\
c\cdot\infty	&= \infty			&&\forall c> 0	
\end{align*}
A parametric integer expression is defined by means of the following grammar:
\begin{align*}
e\;{::=}\; a\mid{\bf p}\mid-{\bf p}\mid{\bf x}\mid e_1\wedge e_2\mid e_1 + e_2\mid e_1\,;\;e_2\mid c\cdot e_1
\end{align*}

\noindent
where $a\in\ZzZ$, $c\in\mathbb N$, ${\bf p}\in{\bf P}$, ${\bf x}\in{\bf X}$.
Given a \emph{parameter setting} $\pi:{\bf P}\to{\mathbb Z}$ and
a \emph{variable assignment} $\xi : {\bf X}\to\ZzZ$, the value of an expression $e$ is determined by:
\begin{align*}
\sem{a}_\pi\,\xi &= a	  &
\sem{{\bf x}}_\pi\,\xi &= \xi({\bf x})	\\
\sem{{\bf p}}_\pi\,\xi &= \pi({\bf p})	&
\sem{e_1\,\Box\,e_2}_\pi\,\xi &= \sem{e_1}_\pi\,\xi\;\Box\;\sem{e_2}_\pi\,\xi	\\
\sem{-{\bf p}}_\pi\,\xi &= -\pi({\bf p})	&
\sem{c\cdot e}_\pi\,\xi &= c\cdot\sem{e}_\pi\,\xi
\end{align*}
Here, $\Box \in \{\wedge, +\}$.
For a given parameter setting $\pi$, 
$\E_\pi$ denotes the (non-parametric) integer equation system obtained from $\E$ by 
replacing every parameter $\bf p$ of $\E$ with its value $\pi({\bf p})$.
For a given parameter setting $\pi$, a \emph{solution} to $\E_\pi$ 
is a variable assignment $\xi^*$ that satisfies all equations of $\E_\pi$.
That is, for each equation ${\bf x} = e_1\vee\ldots\vee e_r$ in $\E$,
\begin{align*}
  \xi^*({\bf x}) = \sem{e_1}_\pi\,\xi^*\vee\ldots\vee\sem{e_r}_\pi\,\xi^*
\end{align*}
Since all operators occurring in right hand sides are monotonic, for every 
parameter setting $\pi$, $\E_\pi$ has a \emph{uniquely determined least solution}.
Finally, a \emph{parametric} solution of $\E$ is a mapping $\Xi$ which assigns to each possible parameter setting
$\pi$, a solution of $\E_\pi$.
$\Xi$ is the \emph{parametric least solution} of $\E$ iff
$\Xi(\pi)$ is the least solution of $\E_\pi$ for every parameter setting $\pi$.

\begin{example}\rm
  Consider the parametric system $\E$ which consists of the single equation $\vx = \mathbf p_1 \vee (\vx + 1 \wedge \mathbf p_2)$. Then the parametric least solution $\Xi$ of $\E$ is given by
  \begin{align*}
    \Xi \; \pi \; \mathbf x &= 
      \begin{cases} 
        \pi(\mathbf p_1) & \text{if } \pi(\mathbf p_1) \geq \pi(\mathbf p_2) \\
        \pi(\mathbf p_2) & \text{if } \pi(\mathbf p_1) < \pi(\mathbf p_2) \\
      \end{cases}
  \end{align*}
    for all parameter settings $\pi$.
  \qed
\end{example}

\section{Parametric Strategy Iteration}\label{s:psi}

In the following we w.l.o.g.\ assume that the set of parameters ${\bf P}$ is given by 
${\bf P} = \{{\bf p}_1,\ldots,{\bf p}_k\}$. 
Accordingly, parameter settings from ${\bf P}\to\ZZ$
can be represented as vectors from $\ZZ^k$.
Our goal is to enhance the strategy iteration algorithm from \us 
to an algorithm 
that computes \emph{parametric least solutions} of systems of parametric integer equations.
Conceptually, we do this by performing each operation for all parameter settings \emph{in parallel}.
For that, we lift the complete lattice $\ZzZ$ of integer values to the set $\ZZ^k\to\ZzZ$
of \emph{parametric} values which is again a complete lattice w.r.t.\
the point-wise extension of the ordering on $\ZzZ$
where the least and greatest elements are given by the functions ${\sf const}_{\neginfty}$ and
${\sf const}_\infty$ mapping each vector of parameters to the constant values $\neginfty$ and $\infty$, respectively.
Accordingly, sub-expressions of right-hand sides are no longer evaluated one by one for each parameter setting.
Instead, each binary operator $\Box$ on $\ZzZ$
is lifted to a binary operator $\Box^*$ on parametric values from $\ZZ^k \to\ZzZ$ 
by defining:
\begin{align*}
  (\phi_1\;\Box^*\;\phi_2)(\pi) &= \phi_1(\pi)\,\Box\,\phi_2(\pi) 
\end{align*}
\noindent
for all parameter settings $\pi \in \Z^k$.
In particular, the lifted maximum operator $\vee^*$ 
equals the least upper bound of the complete lattice $\ZZ^k \to\ZzZ$.
Likewise, scalar multiplication with a non-negative constant $c$ is lifted point-wise from a unary
operator on $\ZzZ$ to a unary operator on $\ZZ^k \to\ZzZ$.
For convenience, we henceforth denote the lifted operators with the same symbols by which we denote the
original operators.

The original system $\E$ of parametric equations over $\ZzZ$ thus can be interpreted as
a system of equations over the domain $\ZZ^k\to\ZzZ$ of \emph{parametric values}. 
For all parametric variable assignments $\rho : {\bf X} \to \Z^k \to \ZzZ$,
expressions $e$ are interpreted as follows:
\begin{align*}
\sem{{\bf p}_i}\;\rho &= {\sf proj}_i	&
\sem{{\bf -p}_i}\;\rho &= {\sf -proj}_i	\\
\sem{a}\;\rho  &= {\sf const}_a   &
\sem{{\bf x}}\;\rho &= \rho({\bf x})  \\
\sem{e_1\,\Box\,e_2}\;\rho &= \sem{e_1}\,\rho\;\Box\;\sem{e_2}\,\rho  &
\sem{c\cdot e}\;\rho &= c\cdot\sem{e}\;\rho    
\end{align*}

\noindent
Here, $\Box$ is a binary operator,
${\sf const}_a$ is a parametric value which maps all arguments to the constant $a$,
and ${\sf proj}_i$ denotes the projection onto the $i$th component of its argument vector.

With respect to this interpretation the least solution $\rho^*$ of $\E$ is a mapping of type
${\bf X}\to\ZZ^k\to\ZzZ$. Let us call $\rho^*$ the \emph{least parametric solution}.
Let $\Xi$ denote the parametric least solution as defined in the last section. Then $\Xi$ and $\rho^*$
are not identical --- but in one-to-one correspondence. By fixpoint induction, it can be verified that:
\[
\Xi\;\pi\;{\bf x} = \rho^*{\bf x}\;\pi 
\]
for all variables ${\bf x}\in{\bf X}$, and all parameter settings $\pi\in\ZZ^k$.
In the same way as the abstract domain ${\ZzZ}$ and the operators on ${\ZzZ}$,
we also lift the notion of a strategy from \cite{DBLP:conf/fm/GawlitzaS08}  to the notion of a 
\emph{parametric} strategy.
For technical reasons, let us assume that the right-hand side for each variable in $\E$
is of the form $a\vee e_1\vee\ldots e_r$ where $a\in\ZzZ$. This can always be achieved, e.g., 
by replacing right-hand sides $e$ which are not of the right format with $\neginfty\vee e$.
A parametric strategy $\sigma$ then assigns to each variable ${\bf x}$, a \emph{parametric choice}.
If the right-hand side for ${\bf x}$ is given by $e_0\vee\ldots\vee e_r$,
$\sigma\,{\bf x}$ maps each parameter setting $\pi\in\ZZ^k$ 
to a natural number in the range $[0,r]$ (signifying one of the argument expressions $e_i$).

Moreover, we need an operator ``${\sf next}$'' which takes a given \emph{parametric strategy} $\sigma$ together with
a \emph{parametric variable assignment} $\rho : {\bf X}\to\ZZ^k\to\ZzZ$ and then, 
for every parameter setting $\pi$, 
switches the choice provided by $\sigma$ whenever required by the evaluation of subexpressions according to 
$\rho$. 
That is, $\sigma' = {\sf next}(\sigma,\rho)$ implies
that the following properties hold for all equations
${\bf x} =e_0\vee\ldots\vee e_r$ and all parameter settings $\pi$:
\begin{enumerate}
\item	$\sem{e_{\sigma'\,{\bf x}\,\pi}}\,\rho\,\pi \geq \sem{e_{i}}\,\rho\,\pi$ for all $i \in \{0,\ldots,r\}$.
\item	If $\sem{e_{\sigma\,{\bf x}\,\pi}}\,\rho\,\pi \geq \sem{e_{i}}\,\rho\,\pi$ for all $i$,
	        then $\sigma'\,{\bf x}\,\pi = \sigma\,{\bf x}\,\pi$.
\end{enumerate}
Note that this operator changes the choice given by the argument strategy $\sigma$ only if a real improvement
is guaranteed.
An operator ``${\sf next}$'' with properties 1) and 2) is a \emph{locally optimal} parametric strategy improvement operator
%
because it chooses for each $\vx$ a best alternative everywhere (relative to $\rho$).
For the correctness of the algorithm it would be sufficient to choose some strategy that is an improvement compared 
to the current strategy at $\rho$.
Finally, we need an operator ``${\sf select}$'' which, based on a parametric choice 
$\phi:\ZZ^k\to\NN$, selects one of the
arguments, i.e., ${\sf select}\;\phi\;(v_0,\ldots,v_r)$ is the parametric value given by:
\begin{align*}
{\sf select}\;\phi\;(v_0,\ldots,v_r)\;\pi &= v_{\phi(\pi)}\,\pi
\end{align*}
\noindent
for all parameter settings $\pi$.

\noindent
With these parametric versions of the corresponding operators used by strategy iteration,
we propose the algorithm in Fig.\ \ref{f:int} for systems of parametric integer equations
with $n$ unknowns.
The resulting algorithm is called \emph{parametric strategy iteration}
or PSI for short.

\begin{figure*}[hbt]

\centering
\fbox{$
	\begin{array}{lll}
	\sigma = \{{\bf x}\mapsto{\sf const}_0\mid {\bf x}\in{\bf X}\};&/\!/&\textrm{initial strategy} \\
	{\bf do}\;\{	\\
	\quad {\bf forall}\;({\bf x}\in{\bf X})\;\;\rho({\bf x}) = {\sf const}_\infty; &/\!/&\textrm{begin BF}	\\
	\quad {\bf for}\;({\bf int}\;i=0; i < n; i{+}{+})	\\
	\quad\quad {\bf forall}\;(({\bf x}=e_0\vee\ldots\vee e_r)\in{\cal E})	\\
	\quad\quad\quad\rho({\bf x}) = {\sf select}\,(\sigma\,{\bf x})\,(\sem{e_0}\,\rho,\ldots,\sem{e_r}\,\rho);	&/\!/&\textrm{end BF} \\
	\quad{\it old}=\sigma;	\\
	\quad\sigma = {\sf next}(\sigma,\rho);		&/\!/&\textrm{strategy improvement}\\ 
	\}\;{\bf while}\;({\it old}\neq\sigma);		&/\!/&\textrm{termination detection}\\
	{\sf output}(\rho);
	\end{array}
	$}
\caption{\label{f:int}\label{f:psi}Parametric strategy iteration for a parametric integer equation system with $n$ unknowns.}

\end{figure*}

\begin{figure*}[hbt]

\centering
\fbox{$
        \begin{array}{lll}
        \sigma = \{{\bf x}\mapsto 0\mid {\bf x}\in{\bf X}\};&\qquad\qquad/\!/&\textrm{initial strategy} \\
        {\bf do}\;\{    \\
        \quad {\bf forall}\;({\bf x}\in{\bf X})\;\;\rho({\bf x}) = \infty; &\qquad\qquad/\!/&\textrm{begin BF}	\\
        \quad {\bf for}\;({\bf int}\;i=0; i < n; i{+}{+})       \\
        \quad\quad {\bf forall}\;(({\bf x}=e_0\vee\ldots\vee e_r)\in{\cal E})   \\
        \quad\quad\quad\rho({\bf x}) = \sem{e_{\sigma\,{\bf x}}}\,\rho;       &\qquad\qquad/\!/&\textrm{end BF} \\
        \quad{\it old}=\sigma;  \\
        \quad\sigma = {\sf next}(\sigma,\rho);          &\qquad\qquad/\!/&\textrm{strategy improvement}\\ 
        \}\;{\bf while}\;({\it old}\neq\sigma);         &\qquad\qquad/\!/&\textrm{termination detection}\\
        {\sf output}(\rho);
        \end{array}
        $}
\caption{\label{f:int0}\label{f:psi0}Ordinary strategy iteration for a non-parametric integer equation system with $n$ unknowns.}

\end{figure*}
\noindent

PSI starts with the initial parametric strategy $\sigma$ mapping each variable and 
parameter setting to the constant $0$,
i.e., it selects for each variable and parameter setting the constant term on the right-hand side.
For a given parametric strategy, the Bellman-Ford algorithm is used to determine the \emph{greatest} solution
(the \emph{for}-loop labeled as BF). This Bellman-Ford iteration amounts to $n$ rounds of round robin iteration 
($n$ the number of variables) starting from the top element of the lattice.
During round robin iteration, the appropriate integer expression $e_i$ as right-hand side for each variable ${\bf x}$ 
and each parameter setting $\pi$, is selected by means of the auxiliary function ${\sf select}$ according to
the current parametric strategy $\sigma$.
As a result of Bellman-Ford iteration for all parameter settings in parallel, 
the next approximation $\rho$ to the least parametric fixpoint is obtained.
This parametric variable assignment then is used to \emph{improve} the current parametric strategy $\sigma$
by means of the operator ``${\sf next}$''. This is repeated until the parametric strategy does not change any more.

For a comparison, Fig.\ \ref{f:int0} shows a version of \emph{the non-parametric} strategy iteration as presented in
\cite{DBLP:conf/fm/GawlitzaS08}.
The mappings $\rho,\sigma$ there have functionalities:
\[
\rho: {\bf X}\to\ZzZ\qquad\sigma:{\bf X}\to\NN
\]
where the evaluation $\sem{e_i}$ of expressions $e_i$ results in integer values only.
Since the strategy $\sigma$ specifies a single integer expression $e_i$ for any given
variable $\bf x$, the call to ``${\sf select}$'' in PSI can be simplified to $\sem{e_{\sigma\,{\bf x}}}\,\rho$.
This optimization is not possible in the parametric case, since different parameter values
may result in different $e_i$ to be selected. 
For systems of integer equalities without parameters
strategy iteration computes the least solution as has been shown in \cite{DBLP:conf/fm/GawlitzaS08}.

Assume for a moment that we can compute with parametric values and parametric strategies effectively,
i.e., can represent them in some data structure, test them for equality,
compute the results of parametric operator applications, as well as realize the operations ``${\sf select}$'' and
``${\sf next}$''. Then the algorithmic scheme from Fig.~\ref{f:int} can be implemented, and we obtain:

\begin{theorem}
Let $\E$ be a parametric integer equation system with $n$ variables
where each right-hand side is a maximum of at most $r$ non-constant integer expressions.
The following holds:
\begin{enumerate}
\item
Parameterized strategy iteration as given by Fig.\ \ref{f:int} terminates after 
at most $(r+1)^n$ strategy improvement steps where each round of improvement requires at most 
${\cal O}(n\cdot|{\cal E}|)$ evaluations of parametric operators.
\item
On termination, the algorithm returns the least parametric solution of $\cal E$.
\end{enumerate}
\end{theorem}

\begin{proof}
First we observe that, when probing the intermediate values of $\sigma$ and $\rho$ for any given parameter setting 
$\pi= (p_1,\ldots,p_k)\in\ZZ^k$, the algorithm
from Fig.\ \ref{f:int} for the parametric system $\E$ 
returns the same values and strategic choices as the algorithm from Fig.\ \ref{f:int0} when run on the integer system $\E_\pi$.
%
Moreover upon termination, the strategic choices $\sigma\,{\bf x}\,\pi$ as well as
the values $\rho\,{\bf x}\,\pi$ do no longer change, and therefore for each parameter setting $\pi$, a least solution
of $\E_\pi$ has been attained.
Since the maximal number of strategies considered by strategy iteration for ordinary integer systems
is bounded by $(r+1)^n$ (independently of the values of the constants in the system), 
we conclude that the parametric algorithm also performs at most $(r+1)^n$ strategy improvement steps.
Therefore, the algorithm terminates.
Since then for each parameter setting $\pi$, a least fixpoint of $\E_\pi$ has been obtained, the resulting
assignment is equal to the least parametric solution of $\cal E$.
\end{proof}
	
\noindent
Being able to effectively compute with parametric variable assignments and 
parametric strategies is crucial for the implementation
of parametric strategy iteration as a practical algorithm.
In the following we will explore
the structure of parametric variable assignments and parametric strategies occurring 
during the algorithm.

A set $S\subseteq\ZZ^k$ of integer points is called \emph{convex} iff 
$S$ equals the set of integer points inside
its convex hull over the rationals.
A mapping $f : \ZZ^k\to V$ ($V$ some set) is called \emph{piecewise constant}
iff there is a finite
partition $\Psi$ of $\ZZ^k$ into nonempty convex sets together with a mapping $\Psi_f : \Psi \to V$ such that
$f(\pi) = \Psi_f(P)$ for all $P \in \Psi$ and all $\pi \in P$.
The cardinality of the partition $\Psi$ is called the \emph{fragmentation} of $f$.
In fact, the fragmentation of $f$ depends on the representation of $f$ rather than the function $f$ itself.
Still, we intentionally do not differentiate between the function and its representation here.
Let ${\sf Aff}_k\subseteq\ZZ^k \to\ZzZ$ denote the set of functions which are
either ${\sf const}_{\neginfty}$, ${\sf const}_\infty$ or an affine function from $\ZZ^k\to\ZZ$.
We call $f \in \ZZ^k \to\ZzZ$ \emph{piecewise affine} iff
there is a piecewise constant mapping 
$\tilde f : \ZZ^k\to{\sf Aff}_k$ such that $f(\pi) = \tilde f (\pi)(\pi)$.
If $f$ is piecewise affine, then there is a finite
partition $\Psi$ of $\ZZ^k$ into non-empty convex sets together with a mapping $\Psi_f:\Psi\to{\sf Aff}_k$ such that
$f(\pi) = \Psi_f(P)(\pi)$ for all $P \in \Psi$ and $\pi \in P$.

Assume $f_1,f_2$ are piecewise affine where both mappings share the same partition $\Psi$ of the
parameter space $\ZZ^k$ into convex sets. 
Then the functions $c\cdot f_1$ as well as $f_1+f_2$
are piecewise affine using the same partition $\Psi$. 
The functions $f_1 \vee f_2$, $f_1 \wedge f_2$, and $f_1;f_2$
are piecewise affine as well with, however, a possibly different finite partition into
convex sets. 
The fragmentation is increased at most by a factor 2.
From that, we conclude that the inner \emph{for}-loop which implements the BF iteration,
may increase the fragmentation of a common partition of values and the parametric strategy $\sigma$
only by a factor $2^{nm_\wedge}$, 
where $m_\wedge$ is the number of occurrences of minimum operators $\wedge$ in $\cal E$.
The applications of the operator ``$;$'' do not contribute, since, in this phase, they never evaluate to a 
parametric value which returns $\neginfty$.
Assume that $\rho$ is the parametric variable assignment computed for 
the parametric strategy $\sigma$ after executing the inner for-loop.
The parametric strategy $\sigma'$ returned by the call to the
${\sf next}$-function for $\sigma$ and $\rho$ will again be a piecewise constant function
whose fragmentation compared with the fragmentation of $\rho$ is increased at most by a factor of 
$2^{m_\vee + m_\wedge + m_;}$,
where $m_\vee$ and $m_;$ denote the number of occurrences of $\vee$-operators and $;$-operators,
respectively.
This holds because we basically have to evaluate right-hand sides in order to 
apply the function $\sf next$ to $\sigma$ and $\rho$.
Therefore, 
compared with the fragmentation of $\sigma$,
the fragmentation of $\sigma'$ increased by a factor of at most 
$2^{m_\vee + m_\wedge + m_;} \cdot 2^{n m_\wedge} = 2^{m_\vee + (n + 1) m_\wedge + m_;}$.
Since the initial strategy has fragmentation 1 and the total number of strategy improvement steps is
bounded, we obtain our second theorem:

\begin{theorem}\label{t:fragmentation}
Consider the parametric strategy improvement algorithm from Fig.\ \ref{f:int}.
All encountered parametric strategies are piecewise constant.
Likewise, all encountered variable assignments are piecewise affine.
Additionally:
\begin{enumerate}
\item	The fragmentation is bounded by $2^{d \cdot (m_\vee + (n + 1) m_\wedge + m_;)}$,
where 
$n$ is the number of unknowns in the integer system of equations,
$m_\Box$ is the number of $\Box$-operators, where $\Box \in \{\vee,\wedge,;\}$, and
$d$ is the maximal number of strategies for any parameter setting. 
\item	The absolute value of any occurring number 
	is bounded by
	$(c\vee 2)^{s n} \cdot a$ where 
	$a$ is the maximum of the absolute values of all constants, 
	$c$ is the maximal occurring constant in a scalar multiplication and 
	$s$ is the maximal size of a right-hand side.
\end{enumerate}
\end{theorem}

\noindent
In the second part of Theorem \ref{t:fragmentation},
we provided bounds for the coefficients occurring in affine functions of parametric
values. These bounds follow since each parametric value is determined by means of $n$ round of
round robin iteration. Consequently the sizes of the numbers occurring in parametric values 
are always polynomial in the input size of PSI. Since each inequality used for refining the current
partition of the parameter space is obtained from the comparison of two affine functions, 
we conclude that the sizes of coefficients of all occurring inequalities also remain
polynomial.


\section{Region Trees}\label{s:region}

The key issue for a practical implementation of PSI 
is to provide an efficient data-structure for partitions of the parameter space $\ZZ^k$ into convex components.
In case $k=1$, i.e., when the system depends on a single parameter only,
the partition $\Psi$ consists of a set of non-empty intervals
$[\neginfty,z_0],[z_0+1,z_1]\ldots, [z_{r-1}+1,z_r],[z_r+1,\infty]$ 
 whose union equals $\ZZ$.
Thus, it can be represented by a finite ordered list $[z_0;\ldots; z_r]$ 
(see Fig.\ \ref{f:list}).
\begin{figure}
\centering
\scalebox{0.5}{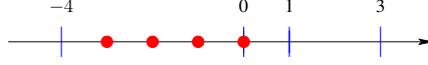}

\caption{\label{f:list}The partition $[-4,0,1,3]$ where the elements of the second region are displayed.}
\end{figure}
By means of the list representation, all required operations on parametric values as well as on 
parametric strategies can be realized in polynomial time.

The case $k>1$ is less obvious. 
We use a representation based on 
satisfiable conjunctions of linear inequalities on parameters $a_1{\bf p}_1+\ldots+a_k{\bf p}_k\leq b$ with $a_1,\ldots,a_k,b\in\ZZ$.
Note that the negation of this inequality is given by the inequality 
$-a_1{\bf p}_1-\ldots -a_k{\bf p}_k\leq - b - 1$.
%
Disjunctions of satisfiable conjunctions of inequalities are organized into a binary tree $t$
as shown, e.g., in Fig.\ \ref{f:tree} on top. Each node in the tree is labeled with
an inequality $c$.
The left child of a node labeled with $c$ corresponds to the case where $c$ holds 
while the right child corresponds to the case where $\neg c$ holds. 
A leaf $v$ of the tree $t$ thus represents the \emph{conjunction} of the inequalities as provided by the
path reaching $v$. 
The path $(c_1,j_1)\ldots(c_r,j_r)$ in $t$ which successively visits the nodes labeled by 
$c_1,\ldots,c_r$ and continues with the $j_1,\ldots,j_r$th successors, respectively,
(where $j_i\in\{1,2\}$), represents the conjunction 
$c_1^{j_1}\wedge\ldots\wedge c_r^{j_r}$ where 
$c^1 = c$ and $c^2 = \neg c$.
As an invariant of $t$, 
we maintain that all conjunctions corresponding to paths in $t$ are
satisfiable.
The leaves of $t$ are annotated with the values attained
in the corresponding parameter region. 
%
%
In order to obtain a more canonical representation, we additionally impose 
a strict linear ordering $\prec$ on the inequalities 
(analogous to the linear ordering of variables for OBDDs)
where the inequality $c$ at a node in $t$ should be less than all successor inequalities. 
Moreover, we demand that $c\prec \neg c$ should hold.
We call the corresponding data-structure \emph{region tree}. 
\begin{example}\label{e:tree}\rm
Consider the following set of linear inequalities:
\begin{align*}
-2{\bf p}_1 + {\bf p}_2	&\leq -2	&
-{\bf p}_1 - {\bf p}_2	&\leq -6	&
-{\bf p}_1 - {\bf p}_2	&\leq -2	
\end{align*}

\noindent
Assume further that the inequalities are ordered from left to right and 
that each of them is smaller than its negation.
Then we obtain a region tree as depicted in Fig.\ \ref{f:tree} at the top, where the integer points corresponding
to the second leaf are shown at the bottom.
The tree is not a full binary tree, since the second inequality $-{\bf p}_1 - {\bf p}_2  \leq -6$
implies the third inequality $-{\bf p}_1 - {\bf p}_2  \leq -2$.
\begin{figure}
\centering
  \raisebox{1cm}{\scalebox{0.6}{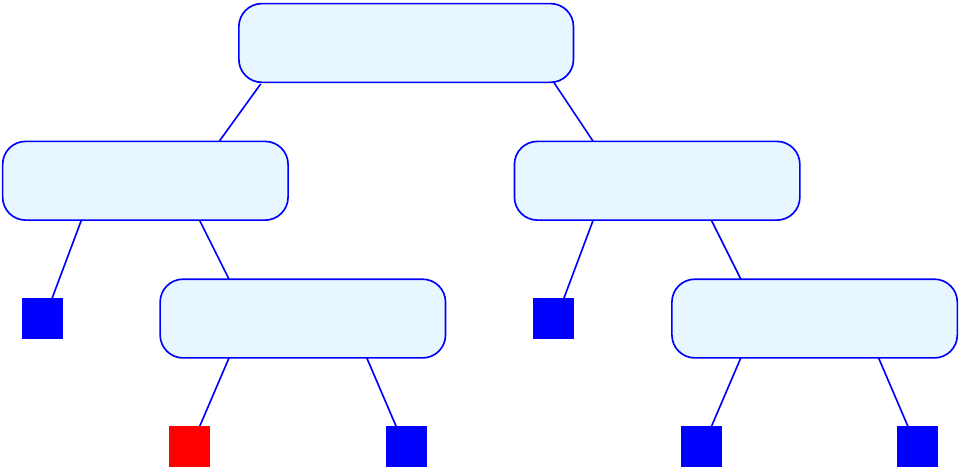}}
	\hspace*{1cm}
  \scalebox{0.7}{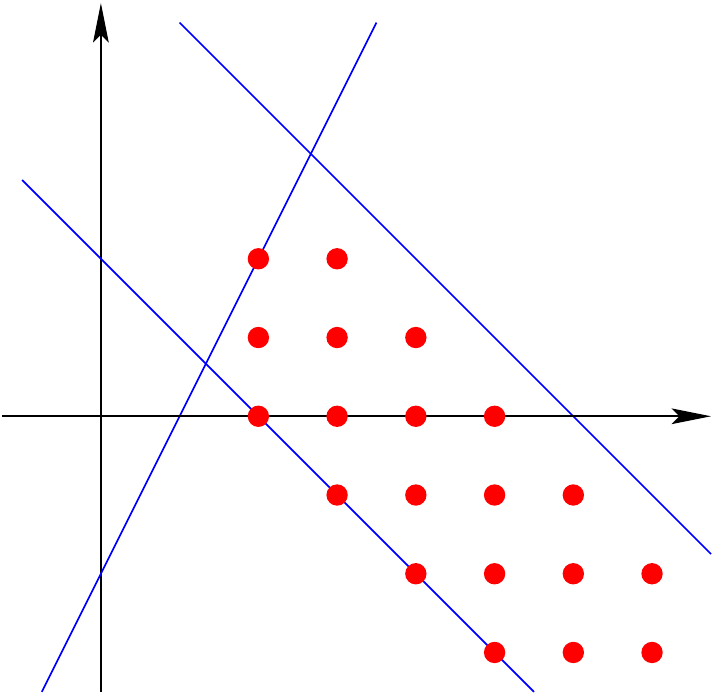}
\caption{\label{f:tree}The region tree for the inequalities of Example \ref{e:tree} together with
the region represented by the second leaf.}
\end{figure}
\qed
\end{example}

\noindent
A trivial upper bound for the number of leaves of a region tree over a set of $n$ linear inequalities 
is given by $2^n$.
However, 
in our application we assume that the parameter space is of fixed dimensionality $k$.
Therefore many occurring inequalities are at least partially redundant.
For the important case where $k \leq n$, we can establish the following 
more precise upper bound on the number of leaves:

\begin{lemma}\label{l:count}
The number of leaves of a region tree over $n$ linear inequalities 
over $k$ parameters with $k \leq n$ is bounded by
$\sum_{i=0}^k \binom n i$.
\qed
\end{lemma}

\noindent
A similar bound has been inferred for cells of maximal dimension $k$ in arrangements of $n$ hyperplanes
(see, e.g., \cite{HalperinCH24}). The intersection of halfspaces as required for our estimation,
results in an identical recurrence relation and therefore in an identical solution.
%
%
%
As a consequence of Lemma \ref{l:count}, the number of leaves and thus also the number of nodes of a region tree for 
a fixed set of parameters is \emph{polynomial} in the number of involved linear inequalities.
%
%

When maintaining region trees, we repeatedly must verify whether a (growing) conjunction
of linear inequalities is satisfiable (over $\ZZ$). Several algorithms have been proposed to solve this
problem (see, e.g., \cite{DBLP:conf/sc/Pught91,dillig09}). If the number of parameters is fixed and small,
polynomial run-time can be guaranteed, e.g., by relying on the LLL algorithm for 
lattices in combination with the ellipsoid method \cite{lll,khachiyan1980polynomial} or by means
of generating functions \cite{Loera_threeinteger}.
Note that for small numbers of variables, even Fourier-Motzkin elimination (though only complete for 
rational satisfiability) is polynomial. 
%
%
%

\section{Implementing Parametric Strategy Iteration}\label{s:implementation}

The efficiency of the resulting algorithm  crucially depends on the fragmentation of parametric variable assignments and
strategies occurring during iteration. 
Instead of globally maintaining a single common partition,
we allow \emph{individual} partitions for each intermediately computed value
as well as for each variable from $\bf X$.
Before applying a binary operator to parametric argument values $t_1,t_2$,
first a common refinement of the partitions of $t_1,t_2$ is computed by means of
a function ``${\sf align}$''. 
Given that 
${'}a\;{\sf tree}$ is the type of region trees whose leaves are labeled with values of type ${'}a$,
the function ``${\sf align}$'' has the following type
\[
\begin{array}{lll}
{\sf align}	&:&	{'}a\;{\sf tree}\to {'}b\;{\sf tree}\to({'}a*{'}b)\;{\sf tree}
\end{array}
\]
In case of addition,
the operator ``$+$'' is applied for each component separately. 
In case of minimum,
the components of the common refinement may be further split into halves 
in order to represent the result as piecewise affine function.
Here, an extra function ``${\sf normalize}$'' is required which re-establishes
the ordering on the inequalities in the tree.
%

Since the number of nodes of a region tree is polynomial in the number $n$ of inequalities, 
and each required subsumption test is polynomial in $n$ and the maximal size of an occurring number,
we obtain:

\begin{lemma}\label{l:pol-1}
Assume that $C$ is a set of $n$ linear inequalities over 
a fixed finite set of parameters where the sizes of all occurring numbers are bounded by $m$.
Then the operations ``${\sf align}$'' as well as addition, scalar multiplication and minimum
lifted to region trees with inequalities from $C$ are polynomial-time in the numbers $n$ and $m$.
\qed
\end{lemma}

\noindent
We build up the nodes of the resulting trees in pre-order.
In order to achieve the given complexity bound, we take care not to introduce nodes which correspond to 
unsatisfiable conjunctions of inequalities, i.e., empty regions. 
Similar to parametric variable assignments, also parametric strategies are not determined as a whole. 
Instead, we maintain for 
each right-hand side $a\vee e_1\vee \ldots\vee e_r$ 
a separate piecewise constant mapping from $\ZZ^k$ into the
range $[0,r]$ of natural numbers which identifies for each parameter setting, the subexpression which is currently selected.
The idea is that one variable ${\bf x}$ should not suffer from the fragmentation
required for another unrelated variable of the system.
Also for the operations ``${\sf next}$'' and ``${\sf select}$'', we obtain:

\begin{lemma}\label{l:pol-2}
Assume that $C$ is a finite set of linear inequalities over 
a fixed finite set of parameters.
Then the operations ${\sf next}$ and ${\sf select}$ 
are polynomial-time in the number $n$ of variables and the size of $C$.
\qed
\end{lemma}

\noindent
Putting lemmas \ref{l:pol-1}, \ref{l:pol-2} together we conclude that
the running time of PSI is fast, whenever the number of occurring inequalities
is small and only few strategies are encountered.
%
%
Summarizing, we obtain:

\begin{theorem}\label{t:pol}
Consider a system $\E$ of integer equations with $k$ parameters.
The least parametric solution of $\E$ can be computed in time
polynomial in the bit size of $\E$, the maximal number of strategies encountered for any parameter setting,
and the maximal number of encountered inequalities.
\qed
\end{theorem}

\noindent
Although in our experiments with interval analysis for the benchmark programs used in Section \ref{s:experiments}, 
the sets of involved inequalities stayed reasonably small,
this however need not always be the case.

\begin{example}\label{e:exp}	\rm
For each $m\geq 0$, consider the following system 
with a single parameter ${\bf p}$:
\begin{align*}
{\bf x}_i &= {\bf x}_{i+1} \vee -2^i + {\bf x}'_{i+1}	&
{\bf x}_m &= {\bf x} \vee -2^m + {\bf x}' &
{\bf x} &= {\bf p} \wedge -{\bf p}	   \\
{\bf x}'_i &= {\bf x}'_{i+1} \wedge 2^i + {\bf x}_{i+1}	&
{\bf x}'_m &= {\bf x}' \wedge 2^m + {\bf x}	&
{\bf x}' &= -{\bf p} \vee {\bf p}	
\end{align*}

\noindent
where $1\leq i<m$. Let $\rho^*$ be the least parametric solution.
Then we have:
\[
\rho^*\;{\bf x}_1\; p =
	\left\{
 	\begin{array}{rll}
 	-p - 2^m	&\quad\text{if}\;& p \leq -2^m -1	\\
 	0		&\quad\text{if}\;& -2^m\leq p\leq 2^m,\; p\;\text{even}	\\
 	-1		&\quad\text{if}\;& -2^m\leq p\leq 2^m,\; p\;\text{odd}	\\
 	p - 2^m		&\quad\text{if}\;& 2^m +1\leq p	\\
 	\end{array}\right.
\]
Thus, the fragmentation of the mapping $\rho^*$ necessarily grows exponentially with $m$.
\qed
\end{example}

\noindent
We conclude that
for large values $m$, any parametric analyzer will exhibit an exponential behavior on the 
system of equations from Example \ref{e:exp}. 

\section{Parametic Program Analysis and Experimental Evaluation}\label{s:experiments}

As indicated in the introduction,
interval analysis for integer variables can be compiled
into a finite system of integer equations. The set of unknowns
of this system are
of the forms ${\bf x}_u^-$ and $ {\bf x}_u^+$ where 
$u$ is a program point 
and 
$x$ is a program variable
of the program to be analyzed, and the superscripts $-,+$ indicate the \emph{negated} lower bounds and
the upper bounds of the respective intervals 
(see \cite{Gawlitza07PreciseFix,DBLP:conf/fm/GawlitzaS08} for the details of the transformation).
The least solution $\rho^*$ of the integer system then translates into the program
invariant which, for program point $u$ and variable $x$, asserts that all runtime values of $x$
are in the interval $[-\rho^*({\bf x}_u^-), \rho^*({\bf x}_u^+)]$.
Here, $[\infty,\neginfty]$ signifies the empty set of values (unreachability of 
$u$).

The transformation of programs into integer equations is readily extended to programs 
\emph{with parameters}. For the sake of the transformation, parameters are treated
as constants occurring in the program, thus resulting in a parametric system of equations
as introduced in Section \ref{s:basics}.
For the program of Example \ref{e:par}, e.g., we obtain:
   \begin{align*}
        {\bf x}_1^- &= {\bf p}_1 \vee {\bf x}_3^-        &
        {\bf x}_2^- &= ({\bf x}_1^+ +\infty);({\bf x}_1^-+{\bf p}_2 -1);({\bf x}_1^-\wedge\infty)      \\
        {\bf x}_3^- &= {\bf x}_2^- + (-1)           &
        {\bf x}_4^- &= ({\bf x}_1^+ +(-{\bf p}_2));({\bf x}_1^-+\infty);({\bf x}_1^- \wedge -{\bf p}_2)      \\
        {\bf x}_1^+ &= {\bf p}_1 \vee {\bf x}_3^+        &
        {\bf x}_2^+ &= ({\bf x}_1^+ +\infty);({\bf x}_1^-+{\bf p}_2 -1);({\bf x}_1^+\wedge {\bf p}_2 -1)      \\
        {\bf x}_3^+ &= {\bf x}_2^+ +1           &
        {\bf x}_4^+ &= ({\bf x}_1^+ +(-{\bf p}_2));({\bf x}_1^-+\infty);({\bf x}_1^+ \wedge \infty)    
      \end{align*}
The least parametric solution 
for the unknowns ${\bf x}_4^-$ and ${\bf x}_4^+$ (signifying
the bounds of the values of $x$ at program exit) is given by:
	\begin{align*}
\xi({\bf x}_4^-) &= {\it if}\; -{\bf p}_1 + {\bf p}_2\leq 0\;{\it then}\;-{\bf p}_1\; {\it else}\;-{\bf p}_2	\\
\xi({\bf x}_4^+) &= {\it if}\; -{\bf p}_1 + {\bf p}_2\leq 0\;{\it then}\;{\bf p}_1\; {\it else}\;{\bf p}_2
	\end{align*}
The resulting \emph{parametric invariant} for the program exit 
states that $x$ equals ${\bf p}_1$, if 
$-{\bf p}_1 + {\bf p}_2\leq 0$, and $x$ equals ${\bf p}_2$ otherwise.

We have provided prototypical implementations of parametric strategy iteration for
parametric 
integer equations, and based on these implementations, also for
parametric 
interval equations.
For convenience, the user may additionally specify 
a boolean combination of linear constraints as 
a global assumption on the parameter values of interest.
Thus, we may, e.g., specify that generally, 
\[
\begin{array}{l}
0\leq {\bf p}_1 \wedge {\bf p}_1 \leq {\bf p}_2 
\end{array}
\]
should hold. Then the analyzer takes 
only considers values less or equal to the tree
in Fig.\ \ref{f:top}.
\begin{figure}
\begin{center}
\scalebox{0.6}{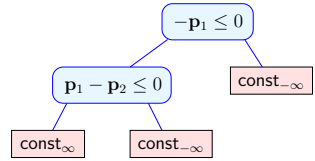}
\caption{\label{f:top}The topmost value under the assumption
	$0\leq {\bf p}_1\leq{\bf p}_2$.}
\end{center}
\end{figure}

One implementation based on lists, deals with the one-parameter case only, while the other implementation, 
which is based on region trees, can deal with multiple parameters.
The total ordering $\prec$ used by our analyzer orders according to the number of variables,
where the lexicographical ordering on the vector of coefficients is used for constraints with the 
same number of variables.
For deciding satisfiability of conjunctions of inequalities, we generally rely on Fourier-Motzkin elimination
(with integer tightening).
Only in the very end, when it comes to produce the final result, we purge regions containing no integer points
by means of an integer solver.
%
%
%
We have tried our implementations on the rate limiter example from \cite{monniaux09} as well as 
on several small (about 20 interval unknowns) but intricate systems of equations
in order to evaluate the impact of the number of parameters as well as 
the impact of the chosen method for checking emptiness of integer polyhedrons
on the practical performance.
%
The tests have been executed on an Intel(R) Core(TM) i5-3427U CPU running Ubuntu.
On that machine, parametric interval analysis of the rate limiter example 
terminated after less than 5s.
%
%
%
%
The remaining benchmarks are based on programs where interval analysis according to the
standard widening/narrowing approach fails to compute the least solution.
For each example equation system, we successively introduce parameters for
the constants used, e.g., in conditions and initializers.
%
%
The system of equations \texttt{nested} is derived from a program with two independent nested loops.
The systems \texttt{amato}$j$ correspond to three example programs presented in
\cite{DBLP:conf/sas/AmatoS13}.
%
The system \texttt{rupak} corresponds to an example program by
Rupak Majumdar presented at MOD'11.
Both \texttt{amato}$2$ and \texttt{rupak} do not realize a plain interval analysis
but additionally track differences of variables.

Interestingly, the number of required strategy improvements does \emph{not} depend
on the number of parameters --- with the notable exception \texttt{amato0} where for three and four parameters,
the number increases from 8 to 9. Generally, the number of iterations is always significantly lower than the
number of unknowns in the system of equations.

\begin{figure*}[hbt]
\centering
\fbox{
%
\includegraphics{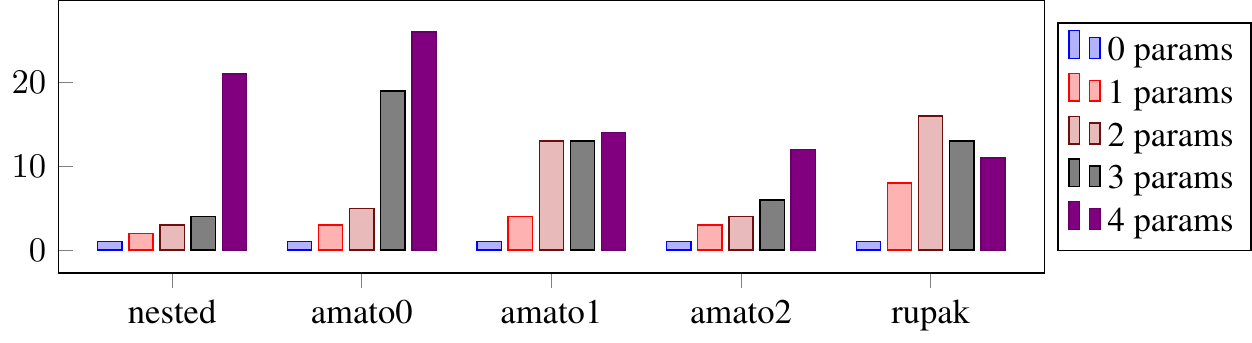}
\hspace{4pt}
}
\caption{Fragmentation.}
\label{fig:regions}
\end{figure*}

%

\noindent
Figure~\ref{fig:regions} shows the number of regions in the results with different behavior. 
For the 0 parameter case this is always 1.
As expected, the fragmentation increases with the number of parameters --- but not as excessively
as we expected. 
%
In case of \texttt{rupak}, the number of regions even \emph{decreased} for three and four parameters.
The reason is that by introducing fresh parameters, also the ordering on inequalities changes.
The ordering on the other hand may have a significant impact onto fragmentation.

\begin{figure*}[hbt] 
\centering
\fbox{
\includegraphics{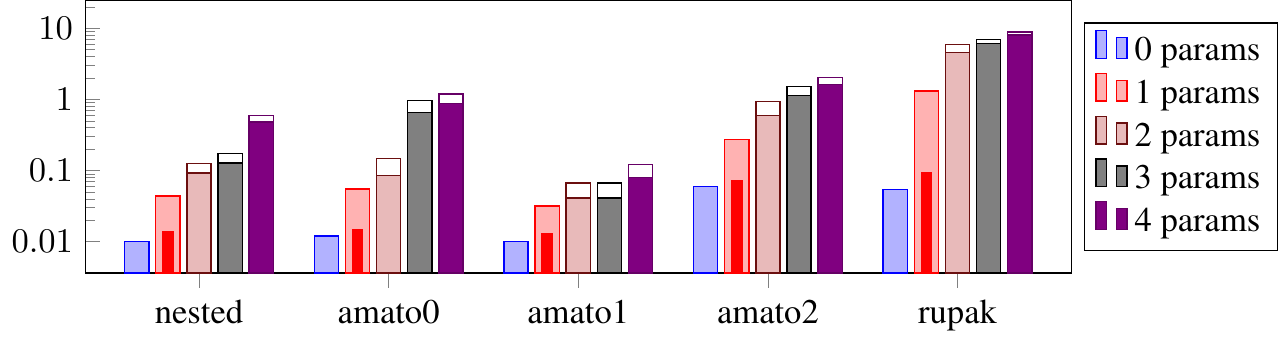}
}
\caption{Execution time in (s).}
\label{fig:time}
\end{figure*}

Figure~\ref{fig:time} shows the running times of the benchmarks on a logarithmic scale. 
We visualize the run-times for 0 through 4 parameters each. The filled and outlined bars correspond to
Fourier-Motzkin elimination and integer satisfiability for testing emptiness of regions, respectively.
The inscribed red bars in the single parameter case represent the run-time 
obtained by using linear lists instead of region trees.
The bottom-line case without parameters is fast and the 
dedicated implementation for single parameters increases the run-time only by a 
factor of approximately 1.4. 
Using region trees clearly incurs an extra penalty, which increases significantly
with the number of parameters. 
While replacing Fourier-Motzkin elimination for testing emptiness of regions 
with an enhanced algorithm for integer satisfiability increases the run-time 
only by an extra factor of about 1.5.
%
When considering absolute run-times, however, it turns out that the solver even for four parameters 
together with full integer satisfiability
is not prohibitively slow (a few seconds only for all benchmark equation systems).
Details on experimental results can be found at \url{www2.in.tum.de/~seidl/psi}.

\section{Related Work}\label{s:related}

Parametric analysis of linear numerical program properties has been advocated by Monniaux 
\cite{monniaux09,monniaux10} by compiling the abstract program semantics to real linear 
arithmetic and then use quantifier elimination to determine the parametric invariants.
We have conducted experiments with Monniaux' tool {\sc Mjollnir},
by which we tried to solve real relaxations of parametric integer equations.
Since {\sc Mjollnir} has no native support for positive or negative infinities 
these values had to be encoded through formulas with extra propositional variables.
This approach, however, did not scale to the sizes we needed.
Our conjecture is that the high Boolean complexity of the formulas causes severe problems.
Beyond that, fewer calls to quantifier elimination for formulas with many variables (as required
by {\sc Mjollnir}) may be more expensive than many calls to an integer solver for formulas
with few variables (namely, the parameters as in our approach).
All in all, since our integer tool  and {\sc Mjollnir}
tackle slightly different problems, a precise comparison is difficult. 
Still, our experiments indicates that our
approach behaves better than a quantifier elimination-based approach for applications where 
the control-flow is complex with multiple control-flow points, but where few parameters are of interest.

Since long, relational program analyses, e.g., by means of polyhedra have been around
\cite{DBLP:conf/popl/CousotH78,DBLP:journals/scp/BagnaraHZ08} which also allow to infer
linear relationships between parameters and program variables. 
The resulting invariants, however, are convex and thus do not allow to differentiate
between different linear dependencies in different regions.
In order to obtain invariants as precise as ours, one would have to combine polyhedral domains
with some form of trace partitioning \cite{mauborgne05trace}.
These kinds of analysis, though, must rely on widening and narrowing to enforce termination,
whereas our algorithms avoid widening and narrowing completely and directly compute 
least solutions, i.e., the best possible parametric invariants.

Parametric analysis of a different kind has also been proposed by Reineke \cite{Reineke14}
in the context of worst-case execution time (WCET). 
They rely on parametric linear programming as implemented by the PIP tool \cite{pip}
and infer the dependence of the WCET on architecture parameters such as the chache size 
by means of black box sampling of the WCETs obtained for different parameter settings.

Our data structure of region trees is a refinement of the tree-like data-structure {\sc Quast}
provided by the PIP tool \cite{pip}. Similar data-structures are also used by Monniaux \cite{monniaux09} to represent
the resulting invariants, and by Mihaila et al.\ \cite{mihaila13widening} to differentiate
between different phases of a loop iteration.
In our implementation, we additionally enforce a total ordering on the 
constraints in the tree nodes and allow arbitrary values at the leaves. 
Total orderings on constraints have also been proposed for linear decision diagrams \cite{DBLP:conf/fmcad/ChakiGS09}.
Variants of LDDs later have been used to implement non-convex linear program invariants 
\cite{Gurfinkel10,Ghorbal12} and in \cite{Gurfinkel11} for representing linear arithmetic formulas when
solving predicate abstraction queries.
In our application, sharing of subtrees is not helpful, since
each node $v$ represents a \emph{satisfiable} conjunction of the inequalities which is constituted by the path
reaching $v$ from the root of the data-structure.
Moreover, our application requires that the leaves of the data-structure are not just annotated with
a Boolean value (as for LDDs), but with values from various sets, namely strategic choices, affine functions 
or even pairs thereof.

\section{Conclusion} \label{s:conc}

Solving systems of \emph{parametric integer equations}
allows to solve also systems of \emph{parametric interval equations},
and thus to realize parametric program analysis
for programs using integer variables.
%
To solve parametric integer equations,
we have presented parametric strategy iteration. 
Instead of solving integer optimization and satisfiability problems
involving all unknowns of the problem formulation (as an approach
based on quantifier elimination),
our algorithm is a smooth generalization of ordinary
strategy iteration, which applies
integer satisfiability to inequalities involving 
parameters only.
%
Our prototypical implementation indicates that this
approach indeed has the potential to deal with nontrivial problems ---
at least when only few parameters are involved.
Introducing further parameters significantly increases the analysis costs.
Surprisingly, the number of strategies required 
as well as the fragmentation observed in our examples increased only moderately.
Accordingly, the required running times were quite decent.



More experiments are necessary, though, to obtain a deeper understanding of parametric strategy iteration.
Also, we are interested in exploring the practical potential of sensitivity and mode
analysis enabled by this new algorithm, e.g., for automotive and avionic
applications.

%
%

\smallskip

\noindent
\emph{Acknowledgement.}
We thank Stefan Barth (LMU) for Example \ref{e:exp}, and
Jan Reineke (Universit\"at des Saarlandes) for useful discussions.


\bibliographystyle{plain}
\bibliography{martin}

\end{document}